\documentclass[year=22]{fmcad}

\usepackage{amsmath}
\usepackage{ifthen}
\usepackage{amssymb}
\usepackage{algpseudocode}
\usepackage{algorithmicx}
\usepackage{algorithm}
\usepackage[english]{babel}
\usepackage{multirow}
\usepackage[utf8]{inputenc}
\usepackage[T1]{fontenc}
\usepackage{flushend}
\usepackage{hyperref}
\usepackage[noabbrev,capitalize]{cleveref}
\usepackage[obeyDraft]{todonotes}
\usepackage{graphicx}

\hypersetup{
  colorlinks=false,
}

\newcommand{\VS}[1]{\mathcal{#1}}

\makeatletter
\@tempcnta=\@ne
\def\@nameedef#1{\expandafter\edef\csname #1\endcsname}
\@whilenum\@tempcnta<27
\do{
  \@nameedef{VS\@Alph\@tempcnta}{\noexpand\mathcal{\@Alph\@tempcnta}}
  \advance\@tempcnta\@ne
}
\makeatother

\newcommand{\St}{~.~}
\newcommand{\Braces}[1]{\left\{#1\right\}}
\newcommand{\Set}[2]{\Braces{#1~|~#2}}
\newcommand{\Size}[1]{\left|#1\right|}
\newcommand{\ApproxSize}[3]{\textsc{MC}_{#2,#3}\pars{#1}}
\newcommand{\Equiv}[2]{\left[#1\right]_{#2}}
\newcommand{\EquivVS}[2]{\left[#1\right]_{\VS{#2}}}
\newcommand{\Relax}[2]{{#1}\rvert_{#2}}
\newcommand{\RelaxVS}[2]{{#1}\rvert_{\VS{#2}}}
\newcommand{\Repl}[3]{#1[#2 \to #3]}

\newcommand{\IfSpecRelax}[2]{\ifthenelse{\equal{#1}{}}{#2}{\Relax{#2}{#1}}}
\newcommand{\IfSpecRelaxVS}[2]{\ifthenelse{\equal{#1}{}}{#2}{\RelaxVS{#2}{#1}}}

\newcommand{\Count}[5][]{\Size{#2\pars{\IfSpecRelax{#1}{#3}, \VS{#4}, \VS{#5}}}}
\newcommand{\CountVS}[5][]{\Size{#2\pars{\IfSpecRelaxVS{#1}{#3}, \VS{#4}, \VS{#5}}}}
\newcommand{\ApproxCount}[7][]{\ApproxSize{#2\pars{\IfSpecRelax{#1}{#3}, \VS{#4}, \VS{#5}}}{#6}{#7}}
\newcommand{\ApproxCountVS}[7][]{\ApproxSize{#2\pars{\IfSpecRelaxVS{#1}{#3}, \VS{#4}, \VS{#5}}}{#6}{#7}}

\newcommand{\counting}[1]{\#{#1}}
\newcommand{\SAT}{\cproblem{SAT}}
\newcommand{\MC}{\counting{\SAT}}
\newcommand{\MaxMC}{\cproblem{Max\MC{}}}
\newcommand{\BaxMC}{\textsc{BaxMC}}
\newcommand{\MaxCount}{\textsc{MaxCount}}
\newcommand{\pars}[1]{\left(#1\right)}
\newcommand{\CEGAR}[0]{\textsc{CEGAR}}

\renewcommand{\P}[1]{\mathbb{P}\left[#1\right]}
\newcommand{\Models}[2]{\mathcal{M}_{\VS{#1}}\pars{#2}}
\newcommand{\cproblem}[1]{\textrm{#1}}

\title{\textsc{BaxMC}: a CEGAR approach to \MaxMC{} \thanks{This work was partially supported by the French ANR project TAVA (ANR-20-CE25-0009) and by the LabEx PERSYVAL-Lab (ANR-11-LABX-0025-01) funded by the French program Investissements d’avenir.
}}
\author{\IEEEauthorblockN{Thomas Vigouroux\IEEEauthorrefmark{1}\orcid{0000-0001-6396-0285}, Cristian Ene\IEEEauthorrefmark{1}\orcid{0000-0001-6322-0383}, David Monniaux\IEEEauthorrefmark{1}\orcid{0000-0001-7671-6126},  Laurent Mounier\IEEEauthorrefmark{1}, Marie-Laure Potet\IEEEauthorrefmark{1}\orcid{0000-0002-7070-6290}}
\IEEEauthorblockA{\IEEEauthorrefmark{1}Univ. Grenoble Alpes, CNRS, Grenoble INP, VERIMAG, 38000 Grenoble, France
 \\ \texttt{Firstname.Lastname@univ-grenoble-alpes.fr}}
}

\usepackage{amsthm}
\theoremstyle{plain}
\newtheorem{lemma}{Lemma}[section]
\newtheorem{theorem}{Theorem}[section]
\newtheorem{corollary}{Corollary}[section]
\theoremstyle{remark}
\newtheorem{remark}{Remark}[section]
\theoremstyle{definition}
\newtheorem{definition}{Definition}[section]
\newtheorem{property}{Property}[section]

\crefname{property}{Property}{Properties}

\newcommand{\doublecolnewline}{\\}

\IEEEoverridecommandlockouts 
\begin{document}

\maketitle

\begin{abstract}
  \MaxMC{} is an important problem with multiple applications in security and program synthesis that
  is proven hard to solve.
  It is defined as: given a parameterized quantifier-free propositional formula, compute parameters such that the number of models of the formula is maximal.
  As an extension, the formula can include an existential prefix.

  We propose a CEGAR-based algorithm and refinements thereof, based on either exact or approximate model counting, and prove its correctness in both cases.
  Our experiments show that this algorithm has much better effective complexity than the state of the art.
\end{abstract}

\section{Introduction}

\MC{} is the problem of counting the solutions of a quantifier-free propositional formula, the counting version of the \SAT{} problem.
\MaxMC{} is the problem of optimizing, according to some propositional variables, the number of solutions according to the others.
We generalize this problem to allow an existential prefix in the formula.

This problem has many practical applications in diverse areas of computer science such as \textit{quantitative program analysis} and \textit{program synthesis} \cite{fremont2017maxcount}.
Most approaches for quantitative information flow analysis use approximations, with fast yet imprecise solutions.
Adaptive attacker synthesis \cite{saha2021incatksynth} would also benefit from advances in \MaxMC{} efficiency, mainly by being able to avoid the use of imprecise heuristics.

Unfortunately, \MaxMC{} has high complexity \cite{Monniaux_MAXSHARPSAT,DBLP:journals/jacm/Toran91}, and practical solving methods remain costly.
At the time of writing, only one solver is publicly available off-the-shelf~\cite{fremont2017maxcount}.

Earlier work on the \MaxMC{} problem proposed two approaches. The first is a
probabilistic solving method \cite{fremont2017maxcount}, which unfortunately degrades to exhaustive search when seeking precise answers to the problem.
The second approach
\cite{audemard2022softcores} solves the problem exactly, but scales poorly.

We present in this paper a new approach to \MaxMC{}, leveraging ideas from \CEGAR{} solvers, and show its effectiveness on various benchmarks used in previous publications on the subject.
We also present improvements of our algorithm based on previous work about symmetry breaking in \SAT{} solvers~\cite{metin2018cdclsym}.

Our contributions are the following:
\begin{itemize}
  \item An \emph{effective algorithm} to compute maximal solutions for the projected model counting problem (\cref{sec:mainalgo,sec:generalization}).
        This algorithm relies either on an \emph{exact} projected model counter as a subprocedure, or on an \emph{approximated} one, which should be the case most times in practice for scalability reasons.
        A complete correctness proof of this algorithm is given for both cases.
  \item The \emph{extension} of our algorithm with SAT \emph{symmetry breaking} techniques
        (\cref{sec:symmetries}) and \emph{heuristics} (\cref{sec:heuristics}), to further improve its efficiency.
  \item The \emph{implementation} of this algorithm in the tool \BaxMC~\cite{baxmc}, together with a
        set of \emph{experimental results} (\cref{sec:evaluation}), showing the accuracy and
        performances of our \MaxMC{} algorithm on various benchmarks, with respect to the only other
        available tool.
\end{itemize}

\section{Preliminaries}

We set our problem in standard Boolean logic.
Throughout the paper, Greek letters ($\phi$, $\psi$, \dots) denote Boolean formulas, uppercase
calligraphic Latin letters ($\VS{V}$, $\VSX$, $\VSY$, $\VSZ$, \dots) denote sets of variables,
simple uppercase Latin letters ($V$, $X$, $Y$, $Z$, \dots) denote variables, lowercase variants of these letters ($x$, $y$, $z$, \dots) denote valuations for these sets of variables.

Let $\mathbb{B} = \left\{ \emph{true}, \emph{false}\right\}$.
A \emph{literal} is a variable or its negation and the set of literals derived from a set of
variables $\VSV$ is denoted by $\overline{\VSV} = \VSV \cup \lnot \VSV$.
Let $\phi(\VS{V})$ be a Boolean formula over $\VSV$ a set of variables.
A valuation $v: \VS{V} \to \mathbb{B}$ is a \emph{model} of $\phi$ if $\phi$ evaluates to \emph{true} over $v$; this is denoted by $v \models \phi$.

We say that a formula $\phi$ is \textit{satisfiable} if there exists $v$ such that $v \models \phi$.
Otherwise, $\phi$ is deemed \textit{unsatisfiable}. Determining whether a formula is satisfiable or
unsatisfiable is called the \textit{satisfiability} problem, also known as \SAT{}.

The \emph{restriction} of a valuation $v: \VS{V} \to \mathbb{B}$ to $\VS{E} \subseteq
  \VS{V}$ is denoted by $\RelaxVS{v}{E}$. We say that two valuations $v_1$ and $v_2$ \textit{agree} on
$\VS{E}$, denoted by $ v_1 \sim_{\VS{E}} v_2$, if their restrictions to $\VS{E}$ are equal.

\subsection{Base definitions}

\begin{definition}[Equivalence class]
  \label{def:equiv}
  Given a valuation $v$ and a set $\VS{E}$, we call \textit{equivalence class of $v$ over $\VS{E}$} the
  set of valuations that agree with $v$ over $\VS{E}$, that is:
  \[
    \EquivVS{v}{E} = \Set{v'}{v' \sim_{\VS{E}} v}
  \]

  We call $\RelaxVS{v}{E}$ \textit{partial} valuations, and $v$
  \textit{complete} valuations.
  The elements of $\EquivVS{v}{E}$ are called the \emph{extensions} or $\RelaxVS{v}{E}$.
\end{definition}

\begin{definition}
  Given propositional formula $\phi(\VS{V})$ and $\VS{E} \subseteq \VS{V}$,  $\Models{E}{\phi} = \Set{\RelaxVS{v}{E}}{v \models \phi}$ denotes the \textit{set
    of models projected over $\VS{E}$}.
\end{definition}

\begin{remark}
  We omit the set $\VS{E}$ when it contains all the variables of $\phi$. That is  $\Models{}{\phi} = \Models{V}{\phi}$.
\end{remark}

\begin{definition}
  Given a valuation $v$, we define the \emph{update} of the variable $V \in \VS{V}$ to $b$ as:
  \[
    \Repl{v}Vb(X) = \left\{ \begin{array}{lr}
      v(X) & \text{if } X \neq V \\
      b & \text{otherwise}
    \end{array}\right.
  \]
\end{definition}

\begin{definition}
  Given two formulas $\phi(\VS{V})$ and $\psi(\VS{V})$, we say that $\phi$
  \emph{entails} $\psi$ (denoted by $\phi \models \psi$) if $\Models{}{\phi} \subseteq
    \Models{}{\psi}$.
\end{definition}

\subsection{Domain-specific definitions}

In the remaining of the paper we consider a \emph{partition} of $\VS{V}$ over three sets $\VS{X}$, $\VS{Y}$ and $\VS{Z}$, respectively
called \textit{witness}, \textit{counting} and \textit{intermediate} variables.
Given a Boolean formula $\phi(\VS{X}, \VS{Y}, \VS{Z})$, we define \MaxMC{} as an optimization problem stated as follows: find
$x_m \in \Models{X}{\phi}$ such that the projected model counting over $\VS{Y}$ of the formula (which will be defined later) is maximal.

\newcommand{\Induced}[3][]{I_{#2}\pars{\IfSpecRelax{#1}{#3}}}
\newcommand{\InducedVS}[3][]{I_{#2}\pars{\IfSpecRelaxVS{#1}{#3}}}
\begin{definition}[Induced set]
  \label{def:induced}
  Given a formula $\phi(\VS{X}, \VS{Y}, \VS{Z})$ and $x \in \Models{X}{\phi}$, the set
  of models over $\VS{Y}$ \textit{induced} by $x$ is:
  \[
    \Induced{\phi}{x} = \Set{y \in \Models{Y}{\phi}}{\exists z \St (x, y, z) \models \phi}
  \]
  We extend this definition to partial witnesses as follows: $\InducedVS[E]{\phi}x = \bigcup_{x' \in
      \EquivVS{x}{E}} \Induced{\phi}{x'}$.
\end{definition}

\begin{definition}[Model counting]
  \label{def:partialprojMC}
  Given a formula $\phi(\VS{X}, \VS{Y}, \VS{Z})$, the \textit{count} of a witness $x$ is defined by the size
  of the set it induces $\Count{\phi}xYZ = \Size{\Induced{\phi}{x}}$.

  We extend this definition to partial witnesses as follows $\CountVS[\VSE]{\phi}xYZ =
    \Size{\InducedVS[E]{\phi}{x}}$.
\end{definition}

\begin{definition}[\MaxMC]
  Given formula $\phi(\VS{X}, \VS{Y}, \VS{Z})$, we can state the \MaxMC{} problem more formally as finding
  $x_m \in \Models{X}{\phi}$ such that:
  \[
    \Count{\phi}{x_m}YZ = \max_{x \in \Models{X}{\phi}}\Count{\phi}xYZ
  \]
\end{definition}

\begin{property}
  \label{prop:partial_bound_complete}
  Given a formula $\phi(\VS{X}, \VS{Y}, \VS{Z})$, the count of a partial witness is an upper-bound of the count of
  its extensions:
  \[
    \forall x' \in \EquivVS{x}{E}, \Count{\phi}{x'}YZ \leq \CountVS[E]{\phi}xYZ
  \]
\end{property}

\begin{proof}
  This follows directly from \cref{def:induced} on induced sets.
\end{proof}

\begin{property}[Monotony of model counting]
  \label{prop:monotonous_count}
  Given a propositional formula $\phi(\VS{X}, \VS{Y}, \VS{Z})$, $\VS{A} \subseteq \VS{B} \subseteq \VS{X}$, and $x \in
    \Models{X}{\phi}$, the count of partial witnesses is monotonous:
  \[
    \CountVS[B]{\phi}xYZ \leq \CountVS[A]{\phi}xYZ
  \]
\end{property}

\begin{proof}
  First, following \cref{def:equiv} we have:
  \[
    \EquivVS{x}{B} \subseteq \EquivVS{x}{A}
  \]

  Hence, following \cref{def:induced}:
  \[
    \InducedVS[B]{\phi}{x} \subseteq \InducedVS[A]{\phi}{x}
  \]

  And that finishes the proof.
\end{proof}

\begin{property}
  \label{prop:relaxing_bounds_each}
  Given a Boolean formula $\phi(\VS{X}, \VS{Y}, \VS{Z})$ such that $x \in \Models{X}{\phi}$, $\VS{E} \subseteq \VS{X}$ and $X_i \in
    \VS{X}$, we have:
  \begin{multline*}
    \CountVS[E]{\phi}xYZ \leq \Count[\VSE - \Braces{X_i}]{\phi}xYZ \leq \\
    \CountVS[E]{\phi}xYZ + \Count{\phi}{\Repl{\RelaxVS{x}{E}}{X_i}{\lnot x(X_i)}}YZ
  \end{multline*}

\end{property}

\begin{proof}
  The first inequality is a direct consequence of \cref{prop:monotonous_count}.
  The last inequality follows from \cref{def:induced}.
\end{proof}

\begin{property}
  \label{prop:phi_goes_down}
  For a given $\phi(\VS{X}, \VS{Y}, \VS{Z})$ and $\phi'(\VS{X}, \VS{Y}, \VS{Z})$ such that $\phi \models \phi'$, a witness $x$, and $\VS{E} \subseteq \VS{X}$ we have:
  \[
    \CountVS[E]{\phi}xYZ \leq \CountVS[E]{\phi'}xYZ
  \]
\end{property}

\begin{proof}
  For any $y \in \InducedVS[E]{\phi}x$, as $y \models \phi$, and $\phi \models \phi'$, we get $y
    \models \phi'$ and hence $\InducedVS[E]{\phi}x \subseteq \InducedVS[E]{\phi'}x$.
\end{proof}

\begin{property}
  \label{prop:stable_psi}
  Given $\phi(\VS{X}, \VS{Y}, \VS{Z})$, $\psi(\VS{X})$ and $x \models \psi$, we have:
  \[
    \Size{\phi(x, \VS{Y}, \VS{Z})} = \Size{(\phi \land \psi)(x, \VS{Y}, \VS{Z})}
  \]
\end{property}

\begin{proof}
  Since $x \models \psi$ and $\psi$ does not depend on $\VS{Y}$ and $\VS{Z}$, $(x, y, z) \models \phi$ if and only if
  $(x, y, z) \models \phi \land \psi$.
\end{proof}

\newcommand{\complexityclass}[1]{\mathsf{#1}}
\newcommand{\Pclass}{\complexityclass{P}}
\newcommand{\nondetclass}[1]{\complexityclass{N}{#1}}
\newcommand{\NP}{\nondetclass{\Pclass}}
\newcommand{\majorityclass}[1]{\complexityclass{P}{#1}}
\newcommand{\PP}{\majorityclass{\Pclass}}
\newcommand{\existsclass}[1]{\exists{#1}}
\newcommand{\sharpP}{\counting{\Pclass}}
\newcommand{\oracleclass}[2]{{#1}^{#2}}
\newcommand{\limitedoracleclass}[3]{{#1}^{#2[#3]}}
\newcommand{\sneg}[1]{\overline{#1}}
\newcommand{\ve}[1]{\mathbf{#1}}
\newcommand{\countmodels}[1]{\Size{\Models{}{#1}}}
\newcommand{\countset}[1]{\Size{{\left\{#1\right\}}}}
\newcommand{\size}[1]{\left|#1\right|}

\section{Solving \MaxMC{}}
\label{sec:mainalgo}

This section presents the main algorithm we propose to solve the \MaxMC{} problem.

\subsection{The main algorithm}

\cref{alg:cegar} takes as input a formula $\phi(\VS{X}, \VS{Y}, \VS{Z})$ and
computes a pair $(x_m, n_m)$ such that $x_m$ is a solution to \MaxMC{} for $\phi$ with model
counting $n_m$. Together with the formula, the algorithm takes multiple precision parameters:
\begin{itemize}
  \item $(\epsilon_i)$ that are called \emph{tolerance} parameters~\cite{chakraborty2013approxmc};
  \item $(\delta_i)$ that are called \emph{confidence} parameters~\cite{chakraborty2013approxmc};
  \item $\kappa$ that is called the \emph{persistence} parameter.
\end{itemize}

Further explanations about these parameters will be given later.

Roughly speaking, this algorithm consists in iterating over possible \emph{witnesses} $x$ of $\phi$.
If the model count for $x$ is less than the current best solution, it
\emph{blocks generalizations} of $x$ such that all extensions of these generalizations are
\emph{worse} than the current best solution (\cref{alg:cegar:generalize,alg:cegar:block}), hence
removing a chunk of the search space at each iteration.
Otherwise, it saves the candidate, which is then the new maximum, and blocks it
(\cref{alg:cegar:newmax,alg:cegar:block}), removing only one candidate from the search space.

\begin{algorithm}[tb]
  \begin{algorithmic}[1]
    \Function{BaxMC$_{\epsilon_0, \epsilon_1, \delta_0, \delta_1, \delta_2,\kappa}$}{$\phi(\VS{X}, \VS{Y}, \VS{Z})$}
    \State $\phi_s \gets \phi$
    \State $x_m \gets \top$
    \State $n_m \gets 0$
    \State $N \gets \ApproxCount{\phi}{\emptyset}YZ{\epsilon_0}{\delta_0}$ \label{alg:cegar:mccall_1}
    \While{$n_m < \frac{N}{1 + \kappa}$}

    \State $x \xleftarrow{\$} \mathcal{M}_X(\phi_s)$ \label{alg:cegar:newsol}
    \Comment{Pick a new candidate}
    \State $c \gets \ApproxCount{\phi_s}xYZ{\epsilon_1}{\delta_1}$ \label{alg:cegar:mccall_2}

    \If{$c > n_m$}
    \Comment{New maximum}
    \State $x_m \gets x$
    \label{alg:cegar:newmax}
    \State $n_m \gets c$
    \State $\VS{E} \gets \VS{X}$
    \Else
    \Comment{Find generalization}
    \State $\VS{E} \gets \Call{Generalize$_{\delta_2}$}{x, \phi_s, n_m}$
    \label{alg:cegar:generalize}
    \EndIf
    \State $\phi_s \gets \phi_s \land \lnot \pars{\RelaxVS{x}{E}}$
    \label{alg:cegar:block}
    \Comment{Block}
    \State $N \gets \ApproxCount{\phi_s}{\emptyset}YZ{\epsilon_0}{\delta_0}$ \label{alg:cegar:mccall_3}
    \EndWhile

    \State \Return $x_m, n_m$
    \EndFunction
  \end{algorithmic}
  \caption{Pseudocode for the \BaxMC algorithm}
  \label{alg:cegar}
\end{algorithm}

We use two kinds of oracles in this algorithm. At \cref{alg:cegar:newsol} we call a \SAT{} solver.
Calls to an existing \MC{} oracle (\cref{alg:cegar:mccall_1,alg:cegar:mccall_2,alg:cegar:mccall_3}) can be performed using either an exact or an approximate model counter.
In the latter case the precision parameters taken as input of the algorithm are used to configure
the oracle, and influence the correctness of the returned value (in the former case, simply assume that they are all~$0$).

\begin{definition}
  \label{prop:n_bound}
  Given $\phi(\VS{X}, \VS{Y}, \VS{Z})$, $x \in \Models{X}{\phi}$ we
  say that $\VS{E} \subseteq \VS{X}$ is $n$-bounding if $\CountVS[E]{\phi}xYZ \leq n$.
\end{definition}

The $\textsc{Generalize}$ function used in \cref{alg:cegar} at \cref{alg:cegar:generalize} is
proved to return $n_m$-bounding sets in both the exact (\cref{thm:gen:correct}) and the
approximate case (with probability $1 - \delta$, \cref{thm:gen:approx}).
The $\textsc{Generalize}$ function is called \cref{alg:elimination} in this paper and it will be presented
in \cref{sec:generalization}.

\subsection{Termination and correctness with an exact \MC{} oracle}

In this subsection, each $i$-indexed variable of the algorithm denotes its value
at the end of the $i$-th iteration of the main loop.
In the exact version of the algorithm, all precision parameters are assumed to be equal to $0$ and all calls to
$\ApproxCount{\phi}xYZ{0}{0}$ return $\Count{\phi}xYZ$.

\label{sec:main:proof}

\begin{theorem}[Termination with an exact \MC{} oracle]
  \cref{alg:cegar} always terminates.
\end{theorem}

\begin{proof}
  By construction of $\pars{{\phi_s}_i}_i$ we have:
  \[
    \forall i > 0, 0 \leq \Size{\Models{X}{{\phi_s}_{i+1}}} < \Size{\Models{X}{{\phi_s}_i}}
  \]
  The sequence $\pars{{n_m}_i}_i$ is obviously increasing. From \cref{prop:phi_goes_down}, the
  sequence $\pars{N_i}_i$ is decreasing and hence $\pars{N_i - {n_m}_i}_i$ is decreasing.

  Putting all this together, $\pars{\Size{\Models{X}{{\phi_s}_i}} + \pars{N_i - {n_m}_i}}_i$ is strictly
  decreasing.

  One can easily see that whenever $\Size{\Models{X}{{\phi_s}_i}} = 0$ it follows that $N_i = 0$ and
  $N_i - {n_m}_i \leq 0$. Hence in all cases, after some iteration $k$, $N_k - {n_m}_k \leq 0$ and
  the termination follows.
\end{proof}

\begin{remark}
  The worst case complexity of \cref{alg:cegar} is reached when it iterates over all the witnesses
  of the formula.
\end{remark}

Let $k$ be the number of iterations performed when \cref{alg:cegar} terminates,
then we have ${n_m}_k \geq N_k$.

\begin{lemma}
  \label{lemma:loop_invariant_complete}
  At every iteration $i$ of \cref{alg:cegar}, we have:
  \[
    \Models{X}{{\phi_s}_i} = \Models{X}{\phi} - \bigcup_{j < i} \Equiv{x_i}{\VS{E}_i}
  \]
\end{lemma}

\begin{proof}
  This follows by construction of ${\phi_s}_i$.
\end{proof}

\begin{lemma}\label{lemma:bounded_union}
  At every iteration $i$ of \cref{alg:cegar}, and assuming $\Call{Generalize}{x, \phi, n}$ returns
  $n$-bounding generalizations of $x$ (as defined in \cref{prop:n_bound}) we have:
  \[
    \forall x' \in \bigcup_{j \leq i} \Equiv{x_j}{\VS{E}_j}, \Count{\phi}{x'}YZ \leq {n_m}_i
  \]
\end{lemma}

\begin{proof}
  Let $j \leq i$, and let $x \in \Equiv{x_j}{\VS{E}_j}$, following \cref{prop:n_bound}, we have
  $\Count{{\phi_s}_j}xYZ \leq {n_m}_j$.

  Then by construction of ${\phi_s}_i$ and \cref{prop:stable_psi} we have $\Count{\phi}xYZ \leq {n_m}_j$
  which, as $\pars{{n_m}_i}_i$ is increasing, proves the lemma.
\end{proof}

\begin{theorem}[Correctness with an exact \MC{} oracle]
  \cref{alg:cegar} is correct, i.e., the returned tuple $\pars{x_m, n_m}$ satisfies the
  following relation:
  \[
    n_m = \Count{\phi}{x_m}YZ = \max_{x \in \Models{X}{\phi}} \Count{\phi}xYZ
  \]
\end{theorem}

\begin{proof}
  Following \cref{prop:partial_bound_complete} and since ${n_m}_k \geq N_k$
  we have:
  \[
    \forall x \in \Models{X}{{\phi_s}_k} \St \Count{\phi}xYZ \leq N_k \leq {n_m}_k
  \]
  Then instantiating \cref{lemma:bounded_union} at iteration $k$ we have:
  \[
    \forall x \in \bigcup_{i \leq k} \Equiv{x_i}{\VS{E}_i}, \Size{\phi(x, Y, Z)} \leq {n_m}_k
  \]
  Following \cref{lemma:loop_invariant_complete}, at iteration $k$ we have
  $\Models{X}{\phi} = \Models{X}{{\phi_s}_k} \cup \bigcup_{i \leq k} \Equiv{x_i}{\VS{E}_i}$, and the result
  follows.
\end{proof}

\subsection{Correctness with a probabilistic \MC{} oracle}

Since the termination can be proven in the same way as in the exact case, we only prove the
correctness.

Let us first recall the expected guarantees provided by an approximate model
counter~\cite{chakraborty2013approxmc}, where the $\epsilon$ parameter characterizes the precision of
the result and the $\delta$ parameter determines its associated confidence.
\begin{property}[Correctness of the Model Counting]
  \label{prop:approx_PAC}
  The count $\ApproxCount{\phi}xYZ\epsilon\delta$ returned by an approximate model counter satisfies
  the following:

  \[
    \P{\frac1{1 + \epsilon} \leq \frac{\ApproxCount{\phi}xYZ\epsilon\delta}{\Size{\phi(x, \VS{Y}, \VS{Z})}} \leq 1 + \epsilon} \geq 1 - \delta
  \]
  These guarantees extend to partial witnesses naturally, i.e., queries of the form
  $\ApproxCount{\phi}{\RelaxVS{x}{E}}{\VS{Y}}{\VS{Z}}\epsilon\delta$.
\end{property}

The next theorem proves the correctness of \cref{alg:cegar} in the approximate case and gives the
associated \emph{tight} bounds.

\begin{theorem}
  \label{thm:cegar:approx}
  Let $(x_m, n_m)$ be the result  returned  by the call \BaxMC{}{$_{\epsilon_0, \epsilon_1,
  \delta_0, \delta_1, \delta_2,\kappa} (\phi(\VS{X}, \VS{Y}, \VS{Z}))$},
  and let
  \[
    M = \max_{x \in \Models{X}{\phi}}\Count{\phi}{x}{Y}{Z}
  \]
  If $\delta_1\leq \frac{\delta_2}{\Size{\VSX}+1}$ then:
  \[
    \P{\frac1{1 + \epsilon_1} \leq \frac{n_m}{\Count{\phi}{x_m}YZ} \leq 1 + \epsilon_1} \geq 1 - \delta_1
  \]
  and
  \begin{multline*}
    \P{\Count{\phi}{x_m}YZ \geq \frac{M}{ (1+\epsilon_0) * (1+\epsilon_1) * (1 +\kappa) } } \\
     \geq (1-\delta_1) * \min(1 - \delta_2, 1- \delta_0)
  \end{multline*}
\end{theorem}

\begin{proof}
  Let $\phi_S$ be the final value of the variable $\phi_s$ after the last iteration of the \textbf{while} loop.
  We have the following guarantees from the approximate model counter (\cref{prop:approx_PAC}):
  \begin{equation}\label{meq3}
    \P{\frac1{1 + \epsilon_1} \leq \frac{n_m}{\Count{\phi}{x_m}YZ} \leq 1 + \epsilon_1} \geq 1 - \delta_1
  \end{equation}
  \begin{equation}\label{meq2}
    \P{\frac1{1 + \epsilon_0} \leq \frac{N}{\Size{\Models{Y}{\phi_S}}} \leq 1 + \epsilon_0} \geq 1 - \delta_0
  \end{equation}

  From \cref{thm:gen:approx} regarding the $\textsc{Generalize}$ function (which will be proved in
  the next section), we also have that for any $x\in  \Models{X}{\phi \wedge \neg \phi_S}$ it holds
  (assuming that $\delta_1\leq \frac{\delta_2}{\Size{\VSX}+1}$):
  \begin{equation}\label{meq4}
    \P{\Count{\phi}{x}YZ \leq n_m} \geq 1 - \delta_2.
  \end{equation}

  After the last iteration of the \textbf{while} loop we have that  $n_m * (1 + \kappa) \geq N$.
  Using this and  \cref{meq2} and \cref{prop:stable_psi} we get that  for any $x\in  \Models{X}{\phi_S}$ it holds
  \begin{eqnarray*}
    \P{\Count{\phi}{x}YZ \leq n_m * (1 + \kappa) * (1+\epsilon_0)}  & \geq & \\
    \P{\Count{\phi}{x}YZ \leq N * (1+\epsilon_0)}  & = & \\
    \P{\Count{\phi_S}{x}YZ \leq N * (1+\epsilon_0)}  & \geq & \\
    \P{ \Size{\Models{Y}{\phi_S}} \leq N * (1+\epsilon_0)} &  \geq & 1 - \delta_0.
  \end{eqnarray*}

  From \cref{meq4}, for any $x\in  \Models{X}{\phi \wedge \neg \phi_S}$ it holds
  \begin{eqnarray*}
    \P{\Count{\phi}{x}YZ \leq n_m  * (1 + \kappa) * (1+\epsilon_0)}  & \geq & \\
    \P{\Count{\phi}{x}YZ \leq n_m}  & \geq & 1 - \delta_2.
  \end{eqnarray*}

  Hence, for any $x\in \Models{X}{\phi}$ it holds
  \begin{multline*}
    \P{\Count{\phi}{x}YZ \leq n_m  * (1 + \kappa) * (1+\epsilon_0)} \doublecolnewline
    \geq \min(1 - \delta_2, 1- \delta_0)
  \end{multline*}

  and hence
  \begin{multline}\label{meq5}
    \P{\frac{n_m}{1+\epsilon_1} \geq \frac{M}{ (1 + \kappa) * (1+\epsilon_0)* (1+\epsilon_1)}} \doublecolnewline
    \geq \min(1 - \delta_2, 1- \delta_0).
  \end{multline}

  Combining this with the \cref{meq3}, we obtain
  \setlength{\belowdisplayskip}{-1em}
  \begin{eqnarray*}
    \P{\Count{\phi}{x_m}YZ \geq \frac{M}{(1 + \kappa) * (1+\epsilon_0) * (1+\epsilon_1) } }  & \geq \\
    (1-\delta_1) * \min(1 - \delta_2, 1- \delta_0).&
  \end{eqnarray*}
\end{proof}

The following corollary instantiates \cref{thm:cegar:approx} in order to get the standard form (as in
\cref{prop:approx_PAC}).

\begin{corollary}
  \label{cor:main:maxcount}
  For any $0< \epsilon, \delta<1$, if in the call of the \BaxMC{} function, we take as parameters
  $\epsilon_0 = \epsilon_1 =\kappa = \sqrt[3]{1 + \epsilon} - 1$, $\delta_0 = \delta_2 =
  \frac{\delta}{2}$ and $\delta_1 = \frac{\delta}{2 * (\Size{\VSX}+1)}$, then the result   $(x_m, n_m)$ satisfies the following inequalities:
  \begin{eqnarray*}
    \P{\frac1{1 + \epsilon} \leq \frac{n_m}{\Count{\phi}{x_m}YZ} \leq 1 + \epsilon} & \geq 1 - \delta \\
    \P{\Count{\phi}{x_m}YZ \geq \frac{M}{1+\epsilon } }  & \geq  1-\delta &
  \end{eqnarray*}
  where
  \[
    M = \max_{x \in \Models{X}{\phi}}\Count{\phi}{x}{Y}{Z}
  \]
\end{corollary}

\begin{proof}
  It is easy to check that $\epsilon_1 = \sqrt[3]{1 + \epsilon} - 1 \leq \epsilon$, $\delta_1 =
  \frac{\delta}{2 * (\Size{\VSX}+1)} <  \delta_2 = \frac{\delta}{2} < \delta$, $(1+\epsilon_0)^3 = 1 + \epsilon$ and $(1-\delta_1) * (1 - \delta_0) = 1 - \delta_0 - \delta_1 + \delta_0 * \delta_1 > 1 - 2 * \delta_0 = 1 - \delta.$
\end{proof}

\todo[inline]{Ajouter ici la version qui dépends de $\kappa$}

\section{Generalization algorithm}
\label{sec:generalization}
\cref{alg:elimination} generalizes a single model $x$ with insufficiently high count to a set of
models with insufficiently high count. This is much the same that a CDCL loop blocks not
only one assignment, but a whole set of assignments.

As shown in \cref{prop:monotonous_count}, generalizing a witness is an instance of the MSMP problem (\textit{Minimal Set subject to a Monotone
  Predicate}), which can be solved using generic algorithms such as \textsc{QuickXPlain}
\cite{junker2001quickxplain}.
Although in theory this should lead to a better algorithm, in practice we observed larger numbers of
calls to the \MC{} oracle, an issue already identified in other contexts~\cite{DBLP:conf/cav/Monniaux10}.

\cref{alg:elimination} is thus a specific solver of the MSMP problem in our setting, relying on a
\emph{linear sweep} over the variables that are part of the valuation.

\begin{algorithm}[tb]
  \begin{algorithmic}[1]
    \Function{Generalize$_{\delta}$}{$x$, $\phi(\VS{X},\VS{Y},\VS{Z})$, $n_m$}

    \State $\VS{E} \gets \VS{X}$
    \ForAll{$X_i \in \VSX$}
    \Comment{Redundancy elimination}
    \If{$\phi\pars{\Repl{x}{X_i}{\lnot x(X_i)}, \VS{Y}, \VS{Z}}$ UNSAT}
    \State $\VS{E} \gets \VS{E} - \left\{X_i\right\}$
    \label{alg:elimination:updateE1}
    \EndIf
    \EndFor

    \State $k \gets \log{n_m} - \log{\ApproxCount[\VSE]{\phi}xYZ\epsilon{\delta_1}}$
    \While{$k > 0 \land \Size{\VS{E}} > 0$} \label{alg:elimination:while}
    \Comment{Log-elimination}
    \State $\VS{A}_k \xleftarrow{\$} \Set{\VS{V} \subseteq \VS{E}}{\Size{\VS{V}} = k}$
    \State $c \gets \ApproxCount{\phi}{x\lvert_{\VS{E} - \VS{A}_k}}YZ\epsilon{\delta_1}$
    \If {$c \leq \frac{n_m}{1 + \epsilon}$}
    \label{alg:elimination:checkbound1}
    \State $\VS{E} \gets \VS{E} - \VS{A}_k$
    \label{alg:elimination:updateE2}
    \State $k \gets \log{n_m} - \log{c}$
    \Else
    \State $k \gets k - 1$
    \label{alg:elimination:bad_log}
    \EndIf
    \EndWhile

    \ForAll {$X_i \in \VSX - \VSE$}
    \Comment{Refinement}
    \If {$\ApproxCount{\phi}{\Relax{x}{\VS{E} - \left\{X_i\right\}}}YZ\epsilon{\delta_1} \leq \frac{n_m}{1 + \epsilon}$}
    \label{alg:elimination:checkbound2}
    \State $\VS{E} \gets \VS{E} - \left\{X_i\right\}$
    \label{alg:elimination:updateE3}
    \EndIf
    \label{alg:elimination:bad_bit}
    \EndFor
    \EndFunction
  \end{algorithmic}
  \caption{Pseudocode for the generalization algorithm}
  \label{alg:elimination}
\end{algorithm}

For efficiency reasons, the steps mentioned in \cref{alg:elimination} are in a
precise order. The reason behind this is:

\begin{enumerate}
  \item The first step relies on a consequence of \cref{prop:relaxing_bounds_each}, allowing to
        relax variables with simple calls to a sat solver.
  \item The log-based generalization is a heuristic allowing to do \textit{big steps} in the
        generalization process by relaxing multiple variables at each loop turn.
  \item The \textit{linear sweep} pass generalizes $x$ in such a way that the returned set is
        minimal, i.e. that none of the further generalizations of the returned value satisfies
        \cref{prop:n_bound}.
\end{enumerate}

The returned $\VSE$ is guaranteed only to be a \emph{local minimum} and it may not be the
\textit{smallest} set such that \cref{prop:n_bound} holds because of the order in which we consider
variables of $\VS{X}$ in \cref{alg:elimination}.

\subsection{Correctness and complexity with an exact \MC{} oracle}

\label{sec:gen:proof}

\begin{property}
  \label{prop:unsat_same}
  If $\phi(\Repl{\RelaxVS{x}{E}}{X_i}{\lnot x(X_i)}, \VS{Y}, \VS{Z})$ is UNSAT then:
  \[
    \CountVS[E]{\phi}xYZ = \Count[\VSE - \left\{ X_i \right\}]{\phi}xYZ
  \]
\end{property}

\begin{proof}
  This follows directly from \cref{prop:relaxing_bounds_each}.
\end{proof}

Let us prove the correctness of \cref{alg:elimination} in the context of an exact \MC{}
oracle. This will finish the correctness proof started in \cref{sec:main:proof}.

\begin{theorem}
  \label{thm:gen:correct}
  \cref{alg:elimination} terminates and is correct:
  the returned set $\VS{E}$ satisfies
  \cref{prop:n_bound}, i.e.,
  $\CountVS[E]{\phi}xYZ \leq n$.
\end{theorem}

\begin{proof}
  In the \textbf{while} loop at \cref{alg:elimination:while} we can see that, at each iteration,
  either $\Size{\VS{E}}$ or $k$ decreases, thus ensuring the termination of the algorithm.

  During any update of the temporary value $\VS{E}$
  (\cref{alg:elimination:updateE1,alg:elimination:updateE2,alg:elimination:updateE3}), we ensure
  that the new value of $\VS{E}$ satisfies \cref{prop:n_bound}:
  \begin{enumerate}
    \item At \cref{alg:elimination:updateE1}, \cref{prop:unsat_same} keeps the model counting stable.
    \item At \cref{alg:elimination:updateE2,alg:elimination:updateE3}, the update is guarded by the
          explicit check of the property (in the \textbf{if} statement
          \cref{alg:elimination:checkbound1,alg:elimination:checkbound2}).
  \end{enumerate}

  Hence the correctness follows.
\end{proof}

\subsection{Bounds with an approximate \MC{} oracle}

\begin{theorem}
  \label{thm:gen:approx}
  Let $\VS{E}\subseteq \VS{X}$ be the set  returned  by the call
  \textsc{Generalize}{$_{\delta}(x, \phi(\VS{X},\VS{Y},\VS{Z}), n)$},
  and assume that
  \[
    \P{\Count{\phi}xYZ \leq n} \geq 1 - \frac{\delta}{\Size{\VS{X}}+1}
  \]

  Then:
  \[
    \P{\Count[\VS{E}]{\phi}{x}YZ \leq n} \geq 1 - \delta
  \]
\end{theorem}

\begin{proof}
  Using \cref{prop:unsat_same}, the variable $\VSE$ after the first loop within
  \cref{alg:elimination} satisfies $\Count{\phi}xYZ =\Count[\VS{E}]{\phi}xYZ$.

  We denote by  $C_{\Relax{x}{\VSV}}$ the value returned by the call
  $\ApproxCount[\VS{V}]{\phi}{x}YZ\epsilon{\delta_1}$.  Since each time we update $\VS{E}$ to a set
  $\VSV$   we ensure $C_{\RelaxVS{x}{V}} \leq \frac{n}{1 + \epsilon}$, we have the following probability:
  \[
    \P{\Count[\VSE]{\phi}xYZ \leq n} \geq 1 - \delta_1
  \]

  Let $\VSE_l$ denote the value obtained after $l$ updates of variable $\VSE$ during
  \textsc{Log-elimination} and \textsc{Refinement} steps within \cref{alg:elimination} and let us
  denote by $P_l$ the probability  that the set $\VSE_l$  is approximately $n$-bounding.

  Using that we update $\VSE$ to the value $\VSE_l$ only  if $C_{\Relax{x}{\VSE_{l}}}*(1+\epsilon)\leq n$,
  we have the following recursive relation:
  \begin{eqnarray*}
    P_l &=& \P{\Count[\VSE_l]{\phi}xYZ \leq n} *  P_{l-1}\\
    & &  \geq (1 - \delta_1)  * P_{l-1} \geq \pars{1 - \delta_1}^l * P_0\\
    & &   \geq \pars{1 - \delta_1}^{l}* \pars{1 - \frac{\delta}{\Size{\VSX} + 1}}
  \end{eqnarray*}

  Thus, as $l \leq  \Size{\VSX}$, if we take $\delta_1 = \frac{\delta}{\Size{\VSX}+1}$ and we call  the
  \MC{} oracle with parameters $(\epsilon, \frac{\delta}{\Size{\VSX}+1})$  we get:
  \[
    \P{\Count[\VSE]{\phi}xYZ \leq n}  \geq \pars{1 - \delta_1}^{l+1} \geq 1 - (l+1) * \delta_1 \geq 1 - \delta
  \]
\end{proof}

\begin{remark}
  The bound with respect to the number of updates is tight. The worst case is reached when the only
  valid subset of $\VSX$ is $\VSX$ itself, that is when the model cannot be generalized.
\end{remark}

\section{Breaking symmetries in \MaxMC{}}
\label{sec:symmetries}
Symmetries are a special kind of permutations of the input variables of a formula leaving it intact.
Exploiting or breaking symmetries in SAT formulas has long been a topic of interest.

For instance, if a formula is left intact by such a permutation then for each blocking
clause $C$, the solver may need to generate the full orbit of $C$ by the group of permutations,
leading to combinatorial explosion. Breaking the symmetry means selecting one solution per orbit by
adding a predicate called \emph{symmetry breaking predicate} to the formula, purposefully generated
to break the symmetries.
The resulting formula is equisatisfiable, but often simpler to solve.

\subsection{Correctness in the presence of symmetries}

In our context, handling symmetries within the witness set reduces the size of the search space, and
leads to better complexity. We give in this section arguments about why this is true.

\begin{definition}
  \label{def:symmetry}
  Given a Boolean formula $\phi(\VSX, \VSY, \VSZ)$, a \textit{symmetry} of $\phi$ is a bijective
  function $\sigma: \overline{\VSX} \mapsto \overline{\VSX}$ that preserves negation,
  that is $\sigma(\lnot X) = \lnot \sigma(X)$, and such that, when $\sigma$ is lifted to formulas,
  $\sigma(\phi) = \phi$ syntactically \cite{zhang2021symmetries}.

  $S_\phi$ denotes the set of all symmetries of $\phi$. We lift $S_\phi$ to models
  by defining the set of symmetries of a model $x$, $S_\phi(x) = \Set{x \circ \sigma}{\sigma \in S_\phi}$.
\end{definition}

\begin{theorem}
  \label{thm:symmetry:correctness}
  In \cref{alg:cegar}, picking only one $x$ per symmetry class of $\phi$ preserves the
  correctness of the algorithm both in the exact and approximate case.
\end{theorem}

\begin{proof}
  \label{rem:symmetry:relaxed_symmetry}
  Whatever the method used to select only one member of each symmetry class, this corresponds to
  creating a symmetry breaking predicate $\psi(\VSX)$ and solving the problem over $\phi \land \psi$,
  and thus the \cref{prop:stable_psi} applies.
\end{proof}

\subsection{Implementing \MaxMC{} symmetry breaking}
We detect symmetries in $\phi$ using the automorphisms of a colored graph representing the formula, defined as follows:
\begin{itemize}
  \item For each variable, create two nodes: one for the positive literal, and one for the negative
    literal. Use color $0$ if the variable is in $X$, otherwise use color $1$. Add an edge (\emph{Boolean consistency edge}) between the
    two nodes.
  \item For each clause, create a node, and assign to it the color~$2$. Add an edge between this clause node
    and every node corresponding to a literal present in the clause.
\end{itemize}

Many tools can be used in order to list the automorphisms of a graph.
In our case, we used
\textsc{bliss} \cite{junttila2007bliss} because of its C++ interface, and its performance.

After detecting the symmetries, one can use any symmetry breaking technique available, either static
\cite{devriendt2016breakid} or dynamic \cite{metin2018cdclsym}. In our implementation, we chose to
use \textsc{CDCLSym} \cite{metin2018cdclsym} because of its ease of use, and because it avoids
generating complex symmetry breaking predicates ahead of time.

\section{Heuristics and optimizations}
\label{sec:heuristics}
We present in this section heuristics used in both \cref{alg:cegar,alg:elimination} in practice, and
discuss their effectiveness.

\subsection{Progressive construction of the candidate}

\label{subsec:progressive_forward}

A simple yet effective optimization is to gradually add literals to the candidate $x$ in
\cref{alg:cegar} at \cref{alg:cegar:newsol}. By stopping earlier, this allows to call
\textsc{Generalize} on a partial assignment instead of a complete one, and will decrease the number of
calls to the \MC{} oracle as it anticipates work that is done in \cref{alg:elimination}.

\subsection{Leads}

\label{subsec:leads}

When performing the generalization in \cref{alg:elimination}, one can see that we can extract
\textit{hints} about promising parts of the search space when relaxing variables.
Indeed, when relaxing parts of the solution
(\cref{alg:elimination:bad_log,alg:elimination:bad_bit}), if the model count of the relaxation goes above
$n_m$, then this part of the search space may contain an improvement over the current solution.

Following this intuition, one can hold a \textit{sorted list}\footnote{The order to use here is:
  first the count of the relaxation, then the size of the relaxation.} of relaxations whose count is above the
current best known maximum, and use it to favor parts of the search space that look promising. We
call these promising relaxations \textit{leads}.
More formally, given $\RelaxVS{\tilde{x}}{E}$ a lead, when searching for a new solution in \cref{alg:cegar} at \cref{alg:cegar:newsol},
instead of searching in $\Models{X}{\phi_s}$, one would search in $\Models{X}{\phi_s} \cap
  \EquivVS{\tilde{x}}{E}$.

Let $L_n\pars{\phi}$ denote the set of leads currently known to the solver with count lower than
$n$. When the currently known maximum is improved in \cref{alg:cegar} at
\cref{alg:cegar:newmax}, we can block all leads whose count is below the new maximum:
\[
  {\phi_s}_{i+1} = {\phi_s}_i \land \bigwedge_{\EquivVS{\tilde{x}}{E} \in L_{n_m}\pars{\phi}} \lnot
  \pars{\RelaxVS{\tilde{x}}{E}}
\]

\subsection{Decision heuristic}

\label{subsec:decisionheurs}

As discussed in \cref{sec:gen:proof}, the performances of the algorithm depend on the order
with which variables are considered in various parts of the solving process (in the
generalization and during the optimization presented in \cref{subsec:progressive_forward}).
One can see that this kind of problem, that we call \textit{variable scheduling},
is actually predominant when solving SAT problems, and even \MC{} problems.

One first heuristic arises from the leads described in \cref{subsec:leads}. One can
use the leads list as indications for literals leading to promising parts of the search space, by
finding the literal which appears the most in the leads. We call this heuristic \texttt{leads}.

Another decision heuristic can be devised using VSIDS~\cite{moskewicz2001chaff}. The idea is to assign a
weight to each literal based on its last appearance in a blocking clause. The weight
of each literal is increased by a constant amount every time the literal appears in a blocking
clause, and is multiplicatively decreased at each blocking clause. This heuristic showed
promising results in both \SAT{} and \counting{\SAT}~\cite{sang2005heuristics}. We call this
heuristic \texttt{vsids}.

One could also choose the next decision variable at random, which we call \texttt{rnd}. And finally,
one could just pick the decision variables in the order they are provided to the tool, which we call
\texttt{none}.

An experimental evaluation is done in \cref{subsec:expe:heur}.

\subsection{Handling equivalent literals}

Equivalent literals are a notorious property of Boolean formulas which, when exploited, results
generally in better runtime performances \cite{lai2021mcequivalent}.

\begin{definition}
  Given a Boolean formula $\phi$, we say that two literals $L_i \in \overline{\VS{V}}$ and $L_j \in
  \overline{\VS{V}}$ are \emph{equivalent} if $\phi \models L_i \Leftrightarrow L_j$.
\end{definition}

Equivalent literals allow to simplify formulas based on the following theorem.
\begin{theorem}
  \label{thm:equiv:simplify}
  Let $\phi$ be a Boolean formula and two equivalent literals $L_i$ and
  $L_j$. Then solving the \MaxMC{} problem for $\phi$ is reduced to solving the \MaxMC{} problem for the simpler
  formula $\phi'$ obtained by replacing all occurrences of $L_j$ (resp. $\lnot L_j$) by $L_i$ (resp.
  $\lnot L_i$) when:

  \begin{enumerate}
    \item either $L_i$ and $L_j$ are in the same literal class (either $\overline{\VSX}$, $\overline{\VSY}$ or
      $\overline{\VSZ}$)
    \item or $L_i \in \overline{\VSX}$ and $L_j \in \overline{\VSY} \cup \overline{\VSZ}$
    \item or $L_i \in \overline{\VSY}$ and $L_j \in \overline{\VSZ}$.
  \end{enumerate}
\end{theorem}

\cref{thm:equiv:simplify} can be applied multiple times in order to further simplify the formula.
Literal equivalence can be detected using binary implication graphs~\cite{heule2011unhiding}.

\section{Experimental evaluation}
\label{sec:evaluation}

\cref{alg:cegar} has been implemented in an open-source tool written in C++ called \BaxMC{}~\cite{baxmc},
including dynamic symmetry breaking techniques (\cref{sec:symmetries}) and all the heuristics discussed in \cref{sec:heuristics}.
In this implementation, we only incorporated the approximate version of the algorithm using
\textsc{ApproxMC5} \cite{meel2020approxmc5} as an approximated model counting oracle and
\textsc{CryptoMinisat} \cite{msoos2009cryptominisat} as a \SAT{} solver oracle. An exact solver is
not implemented because we do not, at the time of writing, have another exact \MaxMC{} solver
available as a comparison.

We use three sets of benchmarks, coming either from~\cite{maxcount}, or from MaxSat 2021 competition~\cite{maxsat2021}.
Benchmarks from this later class are transformed using the method from~\cite{fremont2017maxcount}.
\cref{tab:blist} shows more details about the benchmark set considered. Benchmarks annoted with a
star indicate that a symmetry was found.


All experiments are run on a Dell R640 with 40 cores and 192 GB of RAM running Debian 11, with a
2-hour timeout, a 10 GB memory limit and with parameters $\delta = 0.2$, $\epsilon = 0.8$.

\subsection{Comparison to \MaxCount}

\MaxCount{} \cite{fremont2017maxcount} is used as an off-the-shelf solver of the problem, with
parameters corresponding to $\delta = 0.2$, $\epsilon = 0.8$. Note that these are not the parameters
used in the experiments in \cite{fremont2017maxcount} and that we reimplemented \MaxCount{} using
newer oracles. We did this in order to see how \MaxCount{} and \BaxMC{} behave when both are
providing the same correctness guarantees and using the same oracles for fairness. All figures from
\cref{tab:perf:compare} are obtained when \BaxMC{} is used with the \texttt{(leads,rnd)} heuristic
combination.

\cref{tab:perf:compare} shows the results obtained when running both tools on our three benchmarks.
Bolded values are the \emph{best} values on this line (i.e., smaller time or biggest answer).
The \emph{time} columns are the running times of the tools. The \emph{model count} columns are the
values returned by the candidate tools.

One can see that \BaxMC{} outperforms \MaxCount{} in all benchmark timings. In cases where
\BaxMC{} did not find the best value, it terminates when the bounds on the possible maximum are \textit{tight enough}. This yields a small error margin on the
returned value of \BaxMC{}, but is configurable through its $\kappa$ argument.

\begin{table*}
  \caption{Benchmark list}
  \begin{center}
    \begin{tabular}{|c|c|c|c|c|}
      \hline
      Name                              & $\Size{\VSX}$ & $\Size{\VSY}$ & $\Size{\VSZ}$ & Nr. Clauses \\
      \hline
      \texttt{backdoor-32-24}*          & 32            & 32            & 83            & 76          \\
      \hline
      \texttt{backdoor-2x16-8}*         & 32            & 32            & 136           & 272         \\
      \hline
      \texttt{pwd-backdoor}             & 64            & 64            & 272           & 609         \\
      \hline
      \texttt{bin-search-16}            & 16            & 16            & 1416          & 5825        \\
      \hline
      \texttt{CVE-2007-2875}            & 32            & 32            & 720           & 1740        \\
      \hline
      \texttt{CVE-2009-3002}            & 288           & 240           & 443           & 180         \\
      \hline
      \texttt{reverse}                  & 32            & 32            & 165           & 293         \\
      \hline
      \hline
      \texttt{ActivityService}          & 70            & 34            & 4063          & 15257       \\
      \hline
      \texttt{ActivityService2}         & 70            & 34            & 4063          & 15257       \\
      \hline
      \texttt{ConcreteActivityService}  & 71            & 37            & 4728          & 17856       \\
      \hline
      \texttt{GuidanceService}          & 69            & 27            & 3167          & 11612       \\
      \hline
      \texttt{GuidanceService2}         & 69            & 27            & 3167          & 11612       \\
      \hline
      \texttt{IssueServiceImpl}         & 77            & 29            & 3519          & 13024       \\
      \hline
      \texttt{IterationService}         & 70            & 34            & 4063          & 15257       \\
      \hline
      \texttt{LoginService}             & 92            & 27            & 5110          & 21559       \\
      \hline
      \texttt{NotificationServiceImpl2} & 87            & 32            & 5223          & 22006       \\
      \hline
      \texttt{PhaseService}             & 70            & 34            & 4063          & 15257       \\
      \hline
      \texttt{ProcessBean}              & 166           & 39            & 9675          & 41444       \\
      \hline
      \texttt{ProjectService}           & 134           & 48            & 6778          & 24944       \\
      \hline
      \texttt{sign}                     & 16            & 16            & 107           & 392         \\
      \hline
      \texttt{sign\_correct}            & 16            & 16            & 92            & 346         \\
      \hline
      \texttt{UserServiceImpl}          & 87            & 31            & 3901          & 14653       \\
      \hline
      \hline
      \texttt{drmx}                     & 1030          & 17            & 26            & 2094        \\
      \hline
      \texttt{keller4}                  & 43            & 15            & 62            & 2525        \\
      \hline
      \texttt{g2\_n35e34\_n58e61}       & 34            & 7             & 954           & 38130       \\
      \hline
    \end{tabular}
    \label{tab:blist}
  \end{center}
\end{table*}

\begin{table*}
  \caption{Performance comparisons between \BaxMC{} and \MaxCount}
  \begin{center}
    \begin{tabular}{|c||c|c|c||c|c|}
      \hline
      Benchmark name                    & \multicolumn{3}{|c||}{\textsc{BaxMC}} & \multicolumn{2}{|c|}{\MaxCount}                                                    \\
                                        & Time (s)                              & Sym. Time (s)                   & Model count (log) & Time (s) & Model count (log) \\
      \hline
      \texttt{backdoor-32-24}*          & 611.12                                & \textbf{34.50}                  & \textbf{32}       & 231.87   & \textbf{32}       \\
      \hline
      \texttt{backdoor-2x16-8}*         & \textbf{60.02}                        & 61.07                           & \textbf{16}       & 6512.28  & \textbf{16}       \\
      \hline
      \texttt{pwd-backdoor}             & \textbf{236.87}                       & 240.63                          & \textbf{64}       & TO       & -                 \\
      \hline
      \texttt{bin-search-16}            & 1067.38                               & \textbf{1048.43}                & \textbf{16}       & 1490.44  & \textbf{16}       \\
      \hline
      \texttt{CVE-2007-2875}            & \textbf{36.14}                        & 37.39                           & \textbf{32}       & TO       & -                 \\
      \hline
      \texttt{CVE-2009-3002}            & TO                                    & TO                              & -                 & MO       & -                 \\
      \hline
      \texttt{reverse}                  & TO                                    & TO                              & -                 & MO       & -                 \\
      \hline
      \hline
      \texttt{ActivityService}          & \textbf{3060.39}                      & 3064.60                         & \textbf{33.95}    & TO       & -                 \\
      \hline
      \texttt{ActivityService2}         & 3096.54                               & \textbf{2999.72}                & \textbf{33.95}    & TO       & -                 \\
      \hline
      \texttt{ConcreteActivityService}  & \textbf{84.20}                        & 84.44                           & \textbf{36.91}    & TO       & -                 \\
      \hline
      \texttt{GuidanceService}          & \textbf{1468.39}                      & 1474.51                         & \textbf{26.88}    & TO       & -                 \\
      \hline
      \texttt{GuidanceService2}         & \textbf{1459.74}                      & 1474.76                         & \textbf{26.88}    & TO       & -                 \\
      \hline
      \texttt{IssueServiceImpl}         & 1603.21                               & \textbf{1583.50}                & \textbf{28.88}    & TO       & -                 \\
      \hline
      \texttt{IterationService}         & 3081.86                               & \textbf{3068.95}                & \textbf{33.95}    & TO       & -                 \\
      \hline
      \texttt{LoginService}             & 5275.25                               & \textbf{5197.84}                & \textbf{26.92}    & TO       & -                 \\
      \hline
      \texttt{NotificationServiceImpl2} & \textbf{1286.48}                      & 1287.94                         & \textbf{31.91}    & TO       & -                 \\
      \hline
      \texttt{PhaseService}             & \textbf{3071.86}                      & 3105.18                         & \textbf{33.95}    & TO       & -                 \\
      \hline
      \texttt{ProcessBean}              & TO                                    & TO                              & -                 & TO       & -                 \\
      \hline
      \texttt{ProjectService}           & 5770.26                               & \textbf{5544.02}                & \textbf{47.92}    & TO       & -                 \\
      \hline
      \texttt{sign}                     & 73.56                                 & \textbf{73.43}                  & 15.90             & 819.58   & \textbf{16}       \\
      \hline
      \texttt{sign\_correct}            & 74.58                                 & \textbf{73.78}                  & 15.89             & 819.56   & \textbf{16}       \\
      \hline
      \texttt{UserServiceImpl}          & TO                                    & TO                              & -                 & TO       & -                 \\
      \hline
      \hline
      \texttt{drmx}                     & 24.39                                 & \textbf{24.07}                  & \textbf{16.99}    & TO       & -                 \\
      \hline
      \texttt{keller4}                  & TO                                    & TO                              & -                 & TO       & -                 \\
      \hline
      \texttt{g2\_n35e34\_n58e61}       & \textbf{0.17}                         & 0.41                            & \textbf{2.53}     & TO       & -                 \\
      \hline
    \end{tabular}
    \label{tab:perf:compare}
  \end{center}
\end{table*}

\subsection{Decision heuristic comparison}
\label{subsec:expe:heur}

\cref{tab:perf:heur} shows a comparison between the heuristics that are currently
available in \BaxMC{}.
Lines enumerate the decision heuristics from \cref{subsec:decisionheurs}. Columns specify
heuristics used by the underlying \SAT{} oracle about literals polarities.

Each cell of this table contains, in sequence: the total running time, the number of time
this combination ran the fastest compared to all others, and the number of times this combination
timed out. For example combination \texttt{(leads,cache)} ran for a total time of 62968.38~seconds
with 7~timeouts, and ran the fastest on 3~benchmarks over a total number of 26. In this setup, any
time-out from \BaxMC{} increases the total running time by 7200s.

The table shows that none of the heuristics stands out. We can only eliminate random decision
as a bad heuristic. Nevertheless, the combination of heuristics allows to strongly reduce the overall
number of timeouts.

\begin{table*}
  \caption{Performance comparison between heuristics of \BaxMC{}}
  \begin{center}
    \begin{tabular}{|c|*{4}{|c}|}
      \hline
                               & \texttt{cache}                       & \texttt{neg}                & \texttt{pos}                                  & \texttt{rnd}         \\
      \hline
      \hline
      \texttt{leads}           & 62968.38 -- \textbf{3} -- 7          & 65542.94 -- 2 -- 6          & 65222.40 -- 1 -- \textbf{4}                   & 67233.96 -- 0 -- 5   \\
      \hline
      \texttt{rnd}             & 144081.73 -- 0 -- 19                 & 139002.15 -- 0 -- 18        & 140755.04 -- 0 -- 17                          & 137407.70 -- 0 -- 17 \\
      \hline
      \texttt{none}            & 60368.28 -- \textbf{3} -- 5          & 62729.76 -- 2 -- \textbf{4} & 61317.54 --
      \textbf{3} -- \textbf{4} & 56860.50 -- \textbf{3} -- \textbf{4}                                                                                                      \\
      \hline
      \texttt{vsids}           & 69165.19 -- 2 -- 8                   & 56189.26 -- 1 -- 5          & \textbf{54017.07} -- \textbf{3} -- \textbf{4} & 63865.03 -- 2 -- 6   \\
      \hline
    \end{tabular}
    \label{tab:perf:heur}
  \end{center}
\end{table*}

\section{Related works}
Previous works on \MaxMC{} solving may be classified into three categories, based respectively on probabilistic solving
as in \MaxCount~\cite{fremont2017maxcount}, exhaustive search~\cite{audemard2022softcores} and
knowledge compilation~\cite{oztok2016sddemajsat}.

\emph{Probabilistic solving} relies on ``amplification'' to build a new formula
$\tilde{\phi}(\VS{X}, \VS{Y}, \VS{Z}) = \bigwedge_{i=1}^k \phi(\VS{X}, \VSY_i, \VSZ_i)$, where the
$\VSY_i$ and $\VSZ_i$ are fresh copies of the initial $\VS{Y}$ and $\VS{Z}$ variables, and uniformly
sampling among $\Models{X}{\tilde{\phi}}$.
The higher the $k$, the more the sampling is attracted towards the $\VS{X}$ with large projected model
counting over $(\VSY_i)_{i\leq k}$.
Given parameters $\epsilon$ and $\delta$, the guarantees provided about the returned tuple $(\tilde{n},
  \tilde{x})$ are the same as in \cref{cor:main:maxcount} \cite{fremont2017maxcount}.
Unfortunately, when the size of the formula increases, uniform sampling may become quite expensive as shown in our benchmarks.
Furthermore, this approach is not incremental: looking for a better solution involves re-running the search from scratch.

On the other side of the spectrum lie \emph{exhaustive searches}. The idea here is to make incremental
decisions among the variables in $\VS{X}$, propagating the decision in $\phi$, and simplifying the
formula in order to cache some results \cite{audemard2022softcores}.
Such approaches are exact, but their exhaustive nature limits their scalability.
Component caching \cite{bacchus2003dpllcaching} is a practical way to improve scalability~\cite{audemard2022softcores}
and it could be beneficial into our algorithm too.

\emph{Knowledge compilation} consists in compiling the formula into a representation over which solving the problem (here, the optimal model counting)
is expected to be much easier.
Compilation times tend to dominate and the memory usage of the compiled form may be huge.

A possible approach could use a generalization of $\VS{X}$-constrained SDDs
\cite{oztok2016sddemajsat}. The idea here would be to build $(\VS{X},\VS{Y})$-constrained SDDs, that
is SDDs that are $\VS{X}$-constrained, and for which each subtree that are not over $\VS{X}$ are
$\VS{Y}$-constrained. In this case, one can easily compute the count of every possible pair
$x \in \Models{X}{\phi}$ and then propagate the maximum to the root of the tree. To the best of our
knowledge, this direction has not been explored yet.

\section{Conclusion and future work}

We proposed a CEGAR based algorithm allowing to solve medium-sized instances of the \MaxMC{} within
reasonable time limits, as illustrated in our experiments. This algorithm allows either to compute
exact solutions (when possible), or can be smoothly relaxed to produce approximated results, under
well-defined probabilistic guarantees. Comparisons with an existing probabilistic tool showed the
gains provided by our algorithm on concrete examples. Our implementation and all the related
benchmarks are available on~\cite{baxmc}.

From an algorithmic point of view this work could be extended in several directions.

First, we exploited some classes of symmetries when solving  \MaxMC{} (\cref{sec:symmetries}).
This could be improved by detecting new kinds of symmetries~\cite{devriendt2016breakid}, or exploiting them further
using techniques such as symmetry propagation~\cite{metin2019esbppropag}.

As discussed in \cref{sec:generalization}, our relaxation algorithm (\cref{alg:elimination}) uses a \textit{linear sweep}
over the literals composing a witness.
Instead of returning one possible minimal relaxation,
\textsc{MergeXPlain}~\cite{shchekotykhin2015mergexplain} returns multiple ones, which may be helpful
in our case by allowing the creation of multiple blocking clauses.

As expected, in some instances, our algorithm may degenerate into exhaustive search. While we do not
know yet any characterization of all such instances, we believe that pre-processing and
in-processing~\cite{jarvisalo2012inprocessing} techniques such as
\textsc{unhiding}~\cite{heule2011unhiding} should improve performances and limit the set of inefficient
instances.

Finally, \cref{alg:cegar} may be parallelized by correctly scheduling search
spaces among threads, possibly using the leads described in \cref{subsec:leads}. If we
enforce the fact that all leads currently present in the lead list are disjoint, that is the
$\EquivVS{\tilde{x}}{E}$ are pairwise disjoint (hence splitting the search space into
parts), we expect a favorable parallelization setting.

\bibliographystyle{IEEEtran}
\bibliography{IEEEabrv,biblio}
\end{document}